\documentclass[a4paper,11pt]{article} 
\usepackage[margin=1.2in]{geometry}
\usepackage[english]{babel}

\usepackage[utf8]{inputenc}
\usepackage[ruled,vlined,lined,boxed,commentsnumbered]{algorithm2e}
\usepackage[usenames,dvipsnames,svgnames,table]{xcolor}
\usepackage{wrapfig, blindtext}
\usepackage{amssymb, amsmath} 
\usepackage{mathtools}
\usepackage{tabularx,colortbl}
\usepackage{multirow}
\usepackage{url}

\newtheorem{theorem}{Theorem}
\newtheorem{lemma}{Lemma}
\newtheorem{corollary}{Corollary}
\newtheorem{proof}{Proof}

\newlength{\figurewidth}
\newlength{\smallfigurewidth}

\newcommand{\rbt}{\mathcal B}
\newcommand{\runs}{\mathcal R}

\newcommand{\bigO}{\mathcal{O}}

\begin{document}

\title
{\large
	\textbf{Computing LZ77 in Run-Compressed Space}
}

\author{%
	Alberto Policriti$^{\dag\ast}$ and Nicola Prezza$^{\dag}$\\[0.5em]
	{\small\begin{minipage}{\linewidth}\begin{center}
				\begin{tabular}{ccc}
					$^{\dag}$University of Udine & \hspace*{0.5in} & $^{\ast}$Institute of Applied Genomics \\
					Via delle scienze, 206 && Via J.Linussio, 51 \\
					33100, Udine, Italy && 33100, Udine, Italy\\
					\url{alberto.policriti@uniud.it} &&  \\
					\url{prezza.nicola@spes.uniud.it} &&
				\end{tabular}
			\end{center}\end{minipage}}
		}
		
		\date{}
		
		\maketitle
		\thispagestyle{empty}
		
%

\begin{abstract}
In this paper, we show that the LZ77 factorization of a text $T\in\Sigma^n$ can be computed in $\bigO(\runs\log n)$ bits of working space and $\bigO(n\log\runs)$ time, $\runs$ being the number of runs in the Burrows-Wheeler transform of $T$ reversed. 
For extremely repetitive inputs, the working space can be as low as $\bigO(\log n)$ bits: \emph{exponentially} smaller than the text itself.
As a direct consequence of our result, we show  that a class of repetition-aware self-indexes based on a combination of run-length encoded BWT and LZ77 can be built in asymptotically optimal $\bigO\big(\runs+z)$ words of working space, $z$ being the size of the LZ77 parsing. 
\end{abstract}

\section{Introduction}

Being able to estimate the amount of self-repetitiveness of a text $T\in\Sigma^n$ is a task that stands at the basis of many efficient compression algorithms. While fixed-order statistical methods such as empirical entropy compression are able to exploit only short text regularities~\cite{gagie2006large}, techniques such as Lempel-Ziv parsing (LZ77)~\cite{ziv1977universal}, grammar compression~\cite{charikar2005smallest}, and run-length encoding of the Burrows-Wheeler transform~\cite{burrows1994block,siren2009run,siren2012compressed} have been shown superior in the task of compressing highly repetitive texts. Some recent works showed moreover that such efficient representations can be augmented without asymptotically increasing their space usage in order to support also fast search functionalities~\cite{belazzougui2015composite,claude2011self,kreft2011self} (repetition-aware self-indexes). One of the most remarkable properties of such indexes is the possibility of representing extremely repetitive texts in \emph{exponentially} less space than that of the text itself.

Among the above mentioned repetition-aware compression techniques, LZ77 has been shown to be superior to both grammar-compression~\cite{rytter2003application} and run-length encoding of the Burrows-Wheeler transform~\cite{belazzougui2015composite}. For this reason, much research is focusing into methods to efficiently build, access, and index LZ77-compressed text~\cite{belazzougui2014queries,kreft2011self}. A major concern while building LZ77-based self-indexes is to use small working space. This issue is particularly concerning in situations where the text to be parsed is extremely large and repetitive (e.g. consider the Wikipedia corpus or a large set of genomes belonging to the same species): in such cases, it is not always feasible to load the text into main memory in order to process it, even if the size of the final compressed representation could easily fit in RAM. In these domains, algorithms working in space $\Theta(n\log n)$~\cite{crochemore2008computing}, $\bigO(n\log|\Sigma|)$~\cite{ohlebusch2011lempel,starikovskaya2012computing}, or even $\bigO(nH_k)$~\cite{kreft2011self_thesis,policriti2015fastonline} bits are therefore of little use as they could be exponentially more memory-demanding than the final compressed representation. 

In this work, we focus on the measure of repetitiveness $\runs$: the number of equal-letter runs in the BWT of the (reversed) text. Several works~\cite{belazzougui2015composite,siren2009run,siren2012compressed} studied the empirical behavior of $\runs$ on highly repetitive text collections, suggesting that on such instances $\runs$ grows at the same rate as $z$. Let $\Sigma = \{s_1, ..., s_\sigma\}$ be the alphabet. Both $z$ and $\runs$ are at least $\sigma$ and can be $\Theta(\sigma)$, e.g. in the text $(s_1s_2...s_\sigma)^e$, $e>0$. However, the rate $\runs/z$ can be $\Theta(\log_\sigma n)$: this happens, for example, in de Bruijn sequences (of order $k>1$). 
In this paper, we show how to build the LZ77 parsing of a text $T$ in space bounded by the number $\runs$ of runs in the BWT of $T$ reversed. The main obstacle in computing the LZ77 parsing with a RLBWT index within \emph{repetition-aware space} is the suffix array (SA) sampling: by sampling the SA at regular text positions, this structure requires $\bigO((n/k)\log n)$ bits of working space and supports \emph{locate} queries in time proportional to $k$ (for any $0<k\leq n$). In this work we prove that---in order to compute the LZ77 parsing---it is sufficient to store at most two samples per BWT run, therefore reducing the sampling size to $\bigO(\runs \log n)$ bits. Our algorithm reads the text \emph{only once} from left to right in $\bigO(\log\runs)$ time per character (which makes it useful also in the streaming model). After reading the text, $\bigO(n\log\runs)$ additional time is required in order to output the LZ77 phrases in text-order (the parsing itself is not stored in main memory). The total space usage is $\bigO(\runs\log n)$ bits.

A consequence of our result is that a class of repetition-aware self-indexes combining LZ77 and RLBWT~\cite{belazzougui2015composite} can be built in asymptotically optimal $\bigO(z+\runs)$ words of working space. The only other known repetition-aware index that can be built in asymptotically optimal working space is based on grammar compression and is described in~\cite{yoshimasa2015online}.

\subsection{Notation}

We assume that the text $T\in\Sigma^n$, $\Sigma$ being the alphabet, is terminated by a character $\$\in\Sigma$ not appearing elsewhere in $T$. The LZ77 factorization of $T$ is defined as:
$$\mathcal Z=\langle pos_1,len_1,c_1 \rangle...\langle pos_i,len_i,c_i \rangle...\langle pos_z,len_z,c_z \rangle$$ 
where $0\leq pos_i, len_i <n$, $c_i\in\Sigma$ for $i=1,...,z$, and:
\begin{enumerate}
	\item $T = \omega_1c_1...\omega_zc_z$, with $\omega_i=\epsilon$ if $len_i=0$ and $\omega_i=T[pos_i,...,pos_i+len_i-1]$ otherwise. 
	\item For any $i=1,...,z$ with $len_i>0$, it follows that $pos_i < \sum_{j=1}^{i-1}(len_j+1)$.
	\item For any $i=1,...,z$, $\omega_i$ is the \emph{longest} prefix of $\omega_ic_i...\omega_zc_z$ that occurs at least twice in $\omega_1c_1...\omega_i$
\end{enumerate}

The notation $\overleftarrow S$ indicates the reverse of the string $S\in\Sigma^*$. All BWT intervals are inclusive, and we denote them as pairs $\langle l,r \rangle$ (left-right positions on the BWT). A equal-letter $a$-run in a string $S$ is a substring $W = a^e$, $e>0$ of $S$ such that either (i) $S = W$, (ii) $S = WbX$ or $S = XbW$, $b\in\Sigma,\ b\neq a,\ X\in\Sigma^*$, or (iii) $S = XbWcY$, $b,c\in\Sigma,\ b,c\neq a,\ X,Y\in\Sigma^*$.

A substring $V$ of a string $S\in\Sigma^*$ is \emph{left-maximal} if there exist two distinct characters $a\neq b,\ a,b\in\Sigma$ such that both $Va$ and $Vb$ are substrings of $S$.

\section{Algorithm}\label{sec: alg}

We now describe our algorithm, deferring a detailed description of the employed data structures to the next section.
Let $S = \#T$, $\#\notin\Sigma$ being a character lexicographically smaller than all characters in $\Sigma$. The main structure we use is a dynamic run-length encoded BWT (RLBWT) of the text $\overleftarrow S$. Note that we index strings of the form $\$W\#$, where $W\in (\Sigma-\{\$\})^*$ and $\$$ and $\#$ are the LZ77 and BWT terminators, respectively. The algorithm works in two phases. In the first phase, it reads $S$ from its first to last character, building a RLBWT representation of $\overleftarrow S$. This step employs a well-known online BWT construction algorithm which requires a dynamic string data structure to represent the BWT. The algorithm performs in total $|S|$ \emph{rank} and \emph{insert} operations on the dynamic string (see~\cite{policriti2015fastonline} for a formal description of the procedure). In our case, the dynamic string is also run-length compressed (see the next section for all details). 

In the second phase, the algorithm scans $S$ from left to right by using the BWT just built (i.e. by repeatedly applying LF mapping starting from character $\#$) and outputs the LZ77 factors. We enumerate text positions in $S$-order (despite the fact that we are indexing $\overleftarrow S$). Since characters in $S=\#T$ are right-shifted by 1 position with respect to $T$, we enumerate $S$-positions starting from $-1$, so that $T[i]=S[i],\ 0\leq i < n$ (this simplifies notation).

While reading the $j$-th character of $S$, $j\geq0$, we search in the index the current (reversed) LZ phrase prefix $\overleftarrow{S[i,...,j]},\ j\geq i \geq 0$  ($i$ being the first position of the current phrase). If none of the positions in the BWT interval associated with $\overleftarrow{S[i,...,j]}$ correspond to an occurrence of $S[i,...,j]$ in $S[0,...,j-1]$, then we output the LZ triple $\langle t, j-i, S[j]\rangle$, where $t<i$ and $S[t, ...,t + j-i-1] = S[i, ..., j-1]$ (if $i=j$, then $t=NULL$). The problems to solve are (i) determine whether or not the BWT interval contains occurrences of the phrase prefix in $S[0, ..., j-1]$, and---if the answer is negative and $j>i$---(ii) find $t<i$ such that $S[t, ...,t + j-i-1 ] = S[i, ..., j-1]$. We now show how to answer these queries in $\bigO(\runs \log n)$ bits of working space by keeping in memory $\sigma$ dynamic sets containing in total $\bigO(\runs)$ SA samples.


\subsection{Suffix Array Sampling}

From now on, we will write \emph{BWT} to indicate the Burrows-Wheeler transform of the \emph{whole} string $\overleftarrow S$. Note that, even though we say we \emph{sample the suffix array}, we actually sample text positions associated with BWT positions (i.e. we sample $S$-positions on the $L$ column instead of $S$-positions on the $F$ column of the BWT matrix). Moreover, since we enumerate positions in $S$-order (not $\overleftarrow S$-order), BWT position $k$ will be associated with the sample $n-SA[k]$, $SA[k]$ being the $k$-th entry in the suffix array of $\overleftarrow S$.

Let $0\leq j < n$ be a $S$-position, and $0\leq k < n+1$ be its corresponding position on the BWT. We store SA samples as pairs $\langle j,k\rangle$. Each pair is of one of three types: \emph{singleton}, denoted as $\langle j, k\rangle^\circ$, \emph{open}, denoted as $\prescript{[}{}{\langle j,k\rangle}$, and \emph{close}, denoted as $\langle j,k\rangle^]$. If the pair type is not relevant for the discussion, we simply write $\langle j,k\rangle$. 

Let $\Sigma = \{s_1, ..., s_\sigma\}$ be the alphabet. The samples are stored in $\sigma$ red-black trees $\rbt_{s_1}, ..., \rbt_{s_\sigma}$, and are ordered by BWT coordinate (i.e. the second component of the pairs). While reading character $a=S[j]=BWT[k]$, we first locate the inclusive bounds $l\leq k \leq r$ of its associated BWT $a$-run. We update the trees according to the following rules:

\begin{enumerate}
	\item If no pair $\langle j',k'\rangle\in \rbt_a$ is such that $l\leq k' \leq r$, then we insert the singleton $\langle j, k\rangle^\circ$ in $\rbt_a$. \label{rule1}
	\item If there exists a singleton pair $\langle j',k'\rangle^\circ\in \rbt_a$ such that $l\leq k' \leq r$, then we remove it and:
	\begin{enumerate}
		\item If $k<k'$, then we insert in $\rbt_a$ the pairs $\prescript{[}{}{\langle j,k\rangle}$ and $\langle j',k'\rangle^]$
		\item If $k'<k$, then we insert in $\rbt_a$ the pairs $\prescript{[}{}{\langle j',k'\rangle}$ and $\langle j,k\rangle^]$
	\end{enumerate} \label{rule2}
	\item If there exist two pairs $\prescript{[}{}{\langle j',k'\rangle}, \langle j'',k''\rangle^]\in\rbt_a$ such that $l\leq k' < k'' \leq r$:
	\begin{enumerate}
		\item If $k<k'$, then we remove $\prescript{[}{}{\langle j',k'\rangle}$ from $\rbt_a$ and insert $\prescript{[}{}{\langle j,k\rangle}$ in $\rbt_a$
		\item If $k>k''$, then we remove $\langle j'',k''\rangle^]$ from $\rbt_a$ and insert $\langle j,k\rangle^]$ in $\rbt_a$
		\item Otherwise ($k'<k<k''$), we leave the trees unchanged. 
	\end{enumerate} \label{rule3}
\end{enumerate}

We say that a BWT $a$-run $BWT[l,...,r]$ \emph{contains a pair} or, equivalently, \emph{contains a SA sample} if there exists some $\langle j,k\rangle \in\rbt_a$ such that $l\leq k\leq r$. Moreover, we say that BWT-position $k$ is \emph{marked with a SA sample} if $\langle j,k\rangle\in \rbt_a$, where $a = S[j] = BWT[k]$ ($j$ is the $S$-position corresponding to BWT-position $k$).

It is easy to see that the following invariants hold after the application of any of the above three rules: (i) each BWT run contains either no pairs, a singleton pair, or two pairs---one open and one close; (ii) If a BWT run contains an open $\prescript{[}{}{\langle j',k'\rangle}$ and a close $\langle j'',k''\rangle^]$ pair, then $k'<k''$; (iii) once we add a SA sample inside a BWT run, that run will---from that moment on---always contain at least one SA sample.

By saying that we have \emph{processed $S$-positions} $0,...,j$, we mean that---starting with all trees empty---we have applied the update rules to the SA samples $\langle 0,0 \rangle$, $\langle 1,BWT.LF(0) \rangle$, $\langle 2,BWT.LF^2(0) \rangle, ..., \langle j,BWT.LF^j(0) \rangle$, where $BWT.LF^i(0)$ denotes the LF function applied $i$ times to the BWT-position 0 (e.g. $BWT.LF^2(0) = BWT.LF(BWT.LF(0))$).

We now prove that, after processing $S$-positions $0,...,j$, we can locate at least one occurrence of any string that occurs in $S[0,...,j]$. This property will allow us to locate LZ phrase boundaries and previous occurrences of LZ phrases.

Suppose we have processed $S$-positions $0,...,j$, and let $[l,r]$ be the BWT interval associated with a \emph{left-maximal} string $\overleftarrow V \in \Sigma^*$. The following holds:

\begin{lemma}\label{lemma1}
	There exists a pair $\langle j',k' \rangle\in \rbt_a$ such that $l\leq k' \leq r$ if and only if $Va$ occurs in $S[0,...,j]$.
\end{lemma}

\begin{proof}
$(\Rightarrow)$ If such a pair $\langle j',k' \rangle\in \rbt_a$ exists, where $l\leq k' \leq r$, then clearly $S[j'- m, ..., j'] = Va$. Moreover, since we processed only $S$-positions $0,...,j$, it holds that $j'\leq j$, therefore $Va$ occurs in $S[0,...,j]$.

$(\Leftarrow)$
Let $S[t,...,t+m] = Va$, with $t\leq j-m$. One of the following 2 cases can happen:

(1) The BWT $a$-run containing character $S[t+m] = a$ is a substring of $BWT[l, ..., r]$ \emph{and} is neither a prefix nor a suffix of $BWT[l, ..., r] = Xca^edY$, for some $X,Y\in\Sigma^*,\ c,d\neq a, e>0$. Then---for invariant (iii) and rule \ref{rule1}---since we have visited it (while processing $S$-position $t+m$), the $a$-run must contain at least one SA sample.

(2) The BWT $a$-run containing character $S[t+m] = a$ spans at least two BWT intervals ($\langle l,r \rangle$ included) or is a suffix/prefix of $BWT[l,...,r]$. 
Since  $V$ is left-maximal in $S$, then $BWT[l,...,r]$ contains also a character $b\neq a$. We therefore have that either (i) $BWT[l, ..., r] = a^eXbY$, or (ii) $BWT[l, ..., r] = YbXa^e$, where $X,Y\in\Sigma^*,\ e>0$. The two cases are symmetric, so here we prove only (i). 

Consider all $\overleftarrow S$-suffixes $\overleftarrow{S[0,...,j'']}$ such that 

\begin{itemize}
	\item $j''\leq j$
	\item $Va$ is a suffix of $S[0,...,j'']$
	\item The lexicographic rank of $\overleftarrow{S[0,...,j''-1]}$ among all $\overleftarrow S$-suffixes is $l \leq k'' \leq l+e-1$ (i.e. the suffix lies in $BWT[l,..., l+e-1] = a^e$).
\end{itemize}

There exists at least one such $\overleftarrow S$-suffix: $\overleftarrow{S[0,...,t+m]}$. Then, it is easy to see that the rank $k'$ of the lexicographically largest\footnote{To prove the symmetric case (ii), use the lexicographically \emph{smallest} suffix of this kind.}  $\overleftarrow S$-suffix with the above properties is such that $\langle j',k'\rangle\in\rbt_a$ for some $j'\leq j$. This is implied by the three update rules described above. The BWT position $k$ corresponding to $S$-position $t+m$ lies in the BWT interval $[l, l+e-1]$, therefore either (i) $k$ is the rightmost position visited in its run (thus it is marked with a SA sample), or (ii) the rightmost visited position $k'>k$ in $[l, l+e-1]$ is marked with a SA sample (note that \emph{lexicographically largest} translates to \emph{rightmost} on BWT intervals).
\end{proof}

As a corollary, we note that we can drop the left-maximality requirement from Lemma \ref{lemma1}. Suppose we have processed $S$-positions $0,...,j-1$ (none if $j=0$). The following holds:

\begin{corollary}\label{corollary}
	After processing $S$-positions $j,...,j+m-1$, $m>0$, if a string $W\in\Sigma^m$ occurs in $S[0,...,j+m-1]$, then we can locate one of such occurrences.
\end{corollary}
\begin{proof}
	We prove the property by induction on $|W|=m>0$. Let $W = Va,\ V\in\Sigma^*,\ a\in\Sigma$. If $m=1$, then $V=\epsilon$ (empty string), and the BWT interval associated with $\overleftarrow V$ is the full interval $\langle 0,n \rangle$. But then, $BWT[0,...,n]$ contains at least 2 distinct characters ($a$ and $\#$), so we can apply Lemma \ref{lemma1} to find a previous occurrence of $W=a$.
	
	If $m>1$, then $|V|>0$ and two cases can occur. If the BWT interval associated with $\overleftarrow V$ contains at least 2 distinct characters ($a$ included since by hypothesis $Va$ occurs in $S[0,..., j+m-1]$), then we can apply Lemma \ref{lemma1} to find an occurrence of $W=Va$ in $S[0,..., j+m-1]$. If, on the other hand, the BWT interval associated with $\overleftarrow V$ contains only one distinct character, then this character must be $a$ since $Va$ occurs in $S$. By inductive hypothesis we can locate an occurrence $occ$ of $V$ in $S[0,..., j+m-2]$. But then, since all occurrences of $V$ in $S$ are followed by $a$, $occ$ is also an occurrence of $W=Va$ in $S[0,...,j+m-1]$.
\end{proof}	

%
%
%

\subsection{Pseudocode}

Our complete procedure is reported as Algorithm 1. In line 1 we build the RLBWT of $\overleftarrow S = \overleftarrow{\#T}$ using the online algorithm mentioned at the beginning of this section and employing a dynamic run-length encoded string data structure to represent the BWT. This is the only step that uses the input text, which is read only once from left to right. Since the dynamic string we use is run-length compressed, this step requires $\bigO(\runs\log n)$ bits of working space.

From lines 2 to 9 we initialize all variables that will be used during the algorithm execution. In order: the text length $n$, the current position $j$ on $T$, the position $k$ on $RLBWT$ corresponding to position $j$ on $T$, the current LZ77 phrase prefix length $len$ (last character $T[j]$ excluded), the $T$-position $occ<j$ at which the current phrase prefix $T[j-len,...,j-1]$ occurs ($occ=NULL$ if $len=0$), the red-black trees $\rbt_{s_1}, ..., \rbt_{s_\sigma}$ used to store the SA samples, the current character $c=T[j]=RLBWT[k]$ on the text, and the (inclusive) interval $\langle l,r\rangle$ corresponding to the current reversed LZ phrase prefix $\overleftarrow{T[j-len,...,j-1]}$ on $RLBWT$ ($\langle l,r\rangle$ is the full interval $\langle 0,n\rangle$ if $len=0$). 

The \emph{while} cycle in line 10 scans $T$ positions from the first to last. First of all, we have to discover if the current character $T[j]=c$ ends a LZ phrase. In line 11 we count the number $u$ of runs that intersect interval $[ l,r ]$ on $RLBWT$. 
If $u=1$, then the current phrase prefix $T[j-len,...,j-1]$ is always followed by $c$ in $T$, and consequently $T[j]$ cannot be the last character of the current LZ phrase. Otherwise, by Lemma \ref{lemma1} $T[j-len,...,j]$ occurs in $T[0,...,j-1]$ if and only if there exists a SA sample $\langle j',k'\rangle \in \rbt_c$ such that $l\leq k' \leq r$. The existence of such pair can be verified with a binary search on the red-black tree $\rbt_c$. In line 12 we perform these two tests. If at least one of these two conditions holds, then $T[j-len,...,j]$ occurs in $T[0,...,j-1]$ and therefore is not a LZ phrase. If this is the case, we now have to find $occ<(j-len)$ such that $T[occ, ..., occ+len] = T[j-len,...,j]$ (i.e. a previous occurrence of the current LZ phrase prefix). This goal can be achieved by applying the inductive proof of Corollary \ref{corollary}. If $u=1$ then $occ$ is already the value we need. Otherwise (Lines 13-14) we find a SA sample $\langle j',k'\rangle \in \rbt_c$ such that $l\leq k' \leq r$ (it must exist since $u>1$ and the condition in Line 12 succeeded). Procedure $\rbt_c.locate(l,r)$ returns such $j'$ (to make the procedure deterministic, one could return the value $j'$ associated with the smallest BWT position $l\leq k' \leq r$). Then, we assign to $occ$ the value $j'-len$ (Line 14). We can now increase the current LZ phrase prefix length (Line 15) and update the BWT interval $\langle l,r\rangle$ so that it corresponds to the string $\overleftarrow{T[j-len+1,...,j]}$ (LF mapping in Line 16).

If both the conditions at line 12 fail, then the string $T[j-len,...,j]$ does not occur in $T[0,...,j-1]$ and therefore is a LZ phrase. By the inductive hypothesis of Corollary \ref{corollary}, $occ<j-len$ is either $NULL$---if $len=0$---or such that $T[occ,...,occ+len-1] = T[j-len, ..., j-1]$ otherwise. At line 18 we can therefore output the LZ factor. We now have to open (and start searching in RLBWT) a new LZ phrase: at lines 19-21 we reset the current phrase prefix length, set $occ$ to $NULL$, and reset the interval associated to the current (reversed) phrase prefix to the full interval.

All we have left to do now is to process position $j$ (i.e. apply the update rules to the SA sample $\langle j,k\rangle$) and proceed to the next text position. At line 22 we locate the (inclusive) borders $\langle l_{run}, r_{run} \rangle$ of the BWT run containing position $k$ (i.e. $l_{run} \leq k \leq r_{run}$). This information is used at line 23 to apply the update rules on $\rbt_c$ and on the SA sample $\langle j,k \rangle$. Finally, we increase the current $T$-position $j$ (line 24), compute the corresponding position $k$ on RLBWT (line 25), and read the next $T$-character $c$ on the RLBWT.

\begin{center}
	\begin{figure}
		\includegraphics[trim={2.5cm 10cm 2cm 3cm},clip]{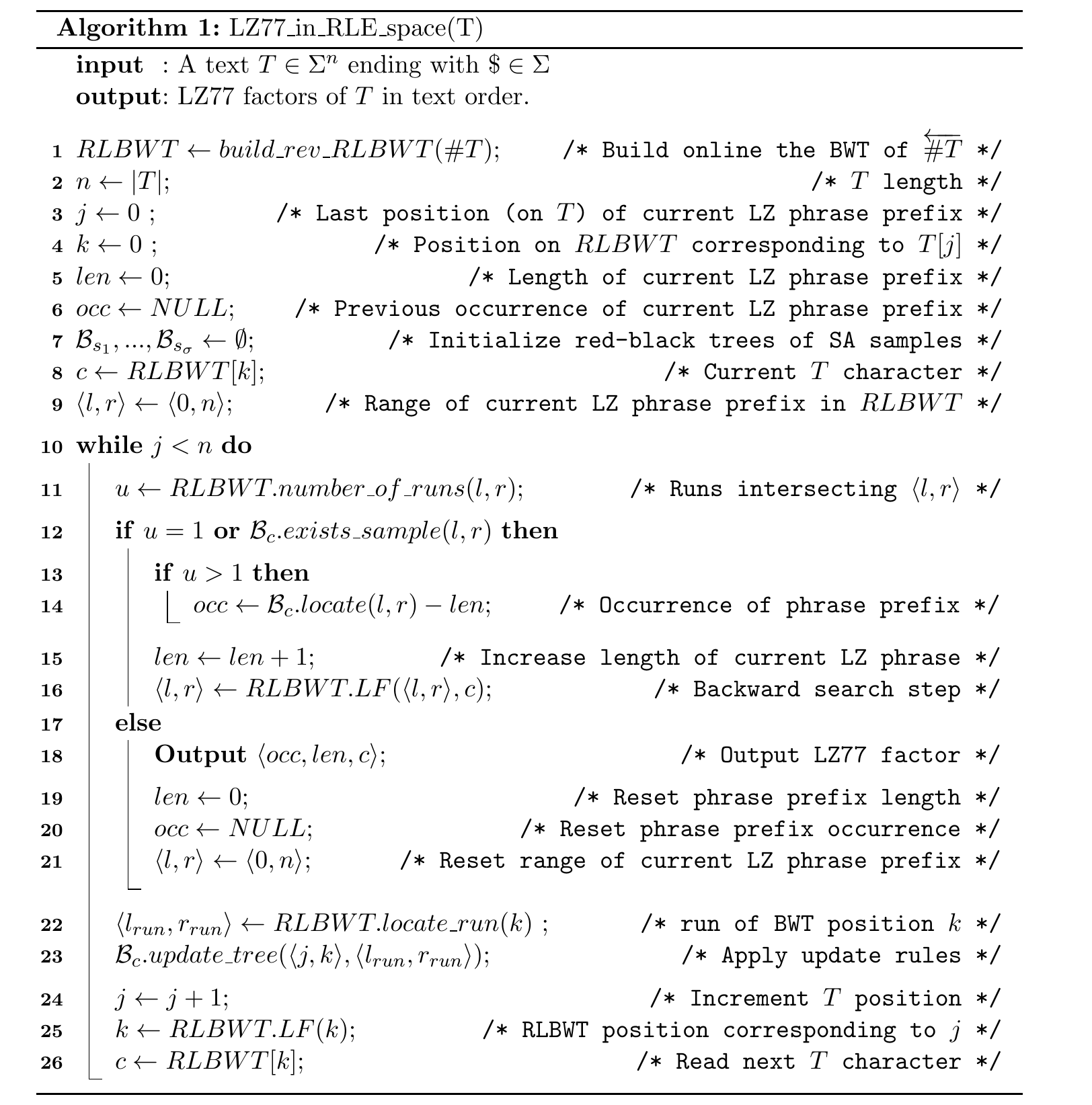}
	\end{figure}
\end{center}

\section{Structures}

To implement the RLBWT data structure, we adopt the approach of~\cite{siren2009run} (RLFM+ index): we store one character per run in a string $H\in (\Sigma\cup\{\#\})^\runs$, we mark the beginning of the runs with a bitvector $V_{all}[0,...,n]$, and we store the lengths of $c$-runs in $\sigma$ bitvectors $V_c$, $c\in\Sigma$ (i.e. length $m$ is encoded as $10^{m-1}$)\footnote{For example, let $BWT=bc\#bbbbccccbaaaaaaaaaaa$. Then, $H=bc\#bcba$, $V_{all}=11110001000110000000000$, $V_a=10000000000$, $V_b = 110001$, and $V_c = 11000$ ($V_\#$ is always $1$).}. In this way, \emph{rank-select-access} operations on the BWT are reduced to \emph{rank-select-access} operations on $H$ and on the bitvectors. If the bitvectors are gap-encoded, the structure takes $\bigO(\runs \log n)$ bits of space. 

The above representation can support also character insertions: an insertion in the run-length string can be easily implemented with a constant number of insertions in $H$ and insertions and 0-deletions---i.e. deleting a 0-bit---in the bitvectors. For space constraints, we do not give these details here.

In order to support insertions, all structures must be dynamic. For $H$ we can use the structure of~\cite{navarro2014optimal} ($\bigO(\runs\log n)$ bits of space and $\bigO(\log\runs)$-time rank, select, access, and insert). We encode the bitvectors with partial sums data structures. Since we want these structures to be dynamic, we are interested in the \emph{Searchable Partial Sums with Indels} (SPSI) problem. This problem consists in maintaining a sequence $s_1, ..., s_m$ of nonnegative integers, each one of $k\in\bigO(w)$ bits, $w$ being the size of a memory word. We want to support the following operations in $\bigO(\log m)$ time with a structure $PS$ of $\bigO(m\log M)$ bits, where $M=\sum_{i=1}^ms_i$:
\begin{itemize}
	\item $PS.sum(i) = \sum_{j=1}^{i} s_j$
	\item $PS.search(x)$ is the smallest $i$ such that $\sum_{j=0}^{i} \geq x$
	\item $PS.update(i,\delta)$: update $s_i$ to $s_i+\delta$. $\delta$ can be negative as long as $s_i+\delta\geq 0$.
	\item $PS.insert(i)$: insert $0$ between $s_{i-1}$ and $s_i$ (if $i=0$, $0$ becomes the first element).
\end{itemize}

We do not need \emph{delete} since we update the bitvectors only with insertions. Once having such a structure, a length-$n$ bitvector $B = 10^{s_1-1}10^{s_2-1}...10^{s_m-1}$ ($s_i>0$) can be encoded in $\bigO(m\log n)$ bits of space with a partial sum $PS$ on the sequence $s_1, ..., s_m$. We need to answer the following queries on $B$: $B[i]$ (\emph{access}), $B.rank(i) = \sum_{j=0}^iB[j]$, $B.select(i)$ (the position $j$ such that $B[j]=1$ and $B.rank(j)=i$), $B.insert(i,b)$ (insert bit $b\in\{0,1\}$ between positions $i-1$ and $i$), and $B.delete_0(i)$, where $B[i]=0$ (delete $B[i]$).

It is easy to see that \emph{rank/access} and \emph{select} operations on $B$ reduce to \emph{search} and \emph{sum} operations on $PS$, respectively. $B.delete_0(i)$ requires just a search and an update on $PS$. To support \emph{insert} on $B$, we can operate as follows. 
$B.insert(i,0),\ i>0$, is implemented with $PS.update( PS.search(i), 1 )$. 
$B.insert(0,1)$ is implemented with $PS.insert(0)$ followed by $PS.update(0,1)$.
$B.insert(i,1),\ i>0$, ``splits'' an integer into two integers: let $j = PS.search(i)$  and $\delta = PS.sum(j)-i$. We first decrease $s_j$ with $PS.update(j,-\delta)$. Then, we insert a new integer $\delta+1$ with the operations $PS.insert(j+1)$ and $PS.update(j+1,\delta+1)$.


\subsection{SPSI implementation}

In our case, the bit length of each integer in the partial sums structures is $k=\log n$, $n$ being the text length. The total number of integers to be stored among the $\sigma+1$ partial sums is $2\runs$. Since we aim at obtaining $\bigO(\runs\log n)$ bits of space, the problem can be solved by using red-black trees (see~\cite{gonzalez2009rank} for a similar construction).

Let $s_1, ..., s_m$ be a sequence of nonnegative integers. We store $s_1, ..., s_m$ in the leaves of a red-black tree, and we store in each internal node of the tree the number of nodes and partial sum of its subtrees. \emph{Sum} and \emph{search} queries can then be easily implemented with a traversal of the tree from the root to the target leaf. \emph{Update} queries require finding the integer (leaf) of interest and then updating $\bigO(\log m)$ partial sums while climbing the tree from the leaf to the root. Finally, \emph{insert} queries require finding an integer (leaf) $s_i$ immediately preceding or following the insert position, substituting it with an internal node with two children leaves $s_i$ and $0$ (the order depending on the insert position---before or after $s_i$), adding 1 to $\bigO(\log m)$ subtree-size counters while climbing the tree up to the root, and applying the RBT update rules. This last step requires the modification of $\bigO(1)$ counters (subtree-size/partial sum) if RBT rotations are involved. All operations take $\bigO(\log m)$ time.

\subsection{Analysis}

It is easy to see that \emph{rank, select, access}, and \emph{insert} operations on RLBWT take $\bigO(\log\runs)$ time each. 
Operations $\rbt_c.exists\_sample(l,r)$ (line 12) and $\rbt_c.locate(l,r)$ (Line 14) require just a binary search on the red-black tree of interest and can also be implemented in  $\bigO(\log\runs)$ time. $RLBWT.number\_of\_runs(l,r)$ is the number of bits set in $V_{all}[l,...,r]$, plus 1 if $V_{all}[l]=0$: this operation requires therefore $\bigO(1)$ rank/access operations on $V_{all}$ ($\bigO(\log\runs)$ time). Similarly, $RLBWT.locate\_run(k)$ requires finding the two bits set preceding and following position $k$ in $V_{all}$ ($\bigO(\log\runs)$ time with a constant number of rank and select operations). We obtain:

\begin{theorem}\label{thm1}
	Algorithm 1 computes the LZ77 factorization of a text $T\in\Sigma^n$ in $\bigO(\runs\log n)$ bits of working space and $\bigO(n\log\runs)$ time, $\runs$ being the number of runs in the Burrows-Wheeler transform of $T$ reversed.
\end{theorem}

As a direct result of Theorem \ref{thm1}, we obtain an asymptotically optimal-space construction algorithm for a class of repetition-aware indexes~\cite{belazzougui2015composite} combining a RLBWT with the LZ77 factorization. 
The construction of such indexes requires building the RLBWT, computing the LZ77 factorization of $T$, and building additional structures of $\bigO(z)$ words of space. We observe that with our algorithm all these steps can be easily carried out in $\bigO(\runs+z)$ words of working space, which is asymptotically otimal.


\bibliographystyle{plain}
\bibliography{rle-lz77.bib}

\end{document}